\documentclass{article}

\usepackage{microtype}
\usepackage{graphicx}
\usepackage{booktabs} 
\usepackage{amsthm}
\usepackage{bbm}
\usepackage{amsmath}
\usepackage{amsfonts}  
\usepackage{bm}   
\usepackage{xcolor}
\usepackage{stmaryrd}
\DeclareMathOperator*{\bigtimes}{\vartimes}
\newtheorem{proposition}{Proposition}

\usepackage[noend]{algorithmic}
\usepackage{subcaption}
\usepackage{caption}
\usepackage{todonotes}
\usepackage{float}

\usepackage{hyperref}

\usepackage{wrapfig}

\definecolor{darkgreen}{RGB}{0,125,0}

\usepackage[accepted]{icml2021}

\icmltitlerunning{Anytime PSRO for Two-Player Zero-Sum Games}

\begin{document}

\twocolumn[
\icmltitle{Anytime PSRO for Two-Player Zero-Sum Games}




\begin{icmlauthorlist}
\icmlauthor{Stephen McAleer}{cmu,dm,uci}
\icmlauthor{Kevin Wang}{uci}
\icmlauthor{John Lanier}{uci}
\icmlauthor{Marc Lanctot}{dm}
\icmlauthor{Pierre Baldi}{uci}
\icmlauthor{Tuomas Sandholm}{cmu}
\icmlauthor{Roy Fox}{uci}
\end{icmlauthorlist}

\icmlaffiliation{cmu}{Carnegie Mellon University}
\icmlaffiliation{uci}{University of California, Irvine}
\icmlaffiliation{dm}{DeepMind}

\icmlcorrespondingauthor{Stephen McAleer}{mcaleer.stephen@gmail.com}

\icmlkeywords{Machine Learning, ICML}

\vskip 0.3in
]



\printAffiliationsAndNotice{}  

\begin{abstract}
\emph{Policy space response oracles (PSRO)} is a multi-agent reinforcement learning algorithm that has achieved state-of-the-art performance in very large two-player zero-sum games. PSRO is based on the tabular \emph{double oracle (DO)} method, an algorithm that is guaranteed to converge to a Nash equilibrium, but may increase exploitability from one iteration to the next. We propose \emph{anytime double oracle (ADO)}, a tabular double oracle algorithm for 2-player zero-sum games that is guaranteed to converge to a Nash equilibrium while decreasing exploitability from one iteration to the next. Unlike DO, in which the restricted distribution is based on the restricted game formed by each player's strategy sets, ADO finds the restricted distribution for each player that minimizes its exploitability against any policy in the full, unrestricted game.  We also propose a method of finding this restricted distribution via a no-regret algorithm updated against best responses, called \emph{RM-BR DO}. Finally, we propose \emph{anytime PSRO (APSRO)}, a version of ADO that calculates best responses via reinforcement learning. In experiments on Leduc poker and random normal form games, we show that our methods achieve far lower exploitability than DO and PSRO and decrease exploitability monotonically.
\end{abstract}

\section{Introduction}
Policy Space Response Oracles is a multi-agent reinforcement learning (RL) method that is based on the tabular double oracle (DO) algorithm~\cite{double_oracle} for finding an approximate \emph{Nash equilibrium (NE)} in two-player zero-sum games. Methods based on PSRO have achieved state-of-the-art performance on large imperfect-information two-player zero-sum games such as Starcraft~\cite{alphastar} and Stratego~\cite{mcaleer2020pipeline}. 

In PSRO, each player maintains, and adds to, a population of policies. 
PSRO can be run until convergence, in which case it outputs a restricted distribution over population policies which corresponds to a NE solution to the original game. In practice, however, PSRO is terminated early in large games. This can be a problem because the PSRO restricted distribution over the population policies is not guaranteed to decrease in exploitability every iteration. As a result, if PSRO is terminated early, the final restricted distribution could even be more exploitable than the initial one.

In this paper, we propose a new PSRO variant that, in each iteration, finds the least-exploitable restricted distribution over the population policies of each player. We present two versions of our algorithm: an exact version (called \emph{anytime double oracle (ADO)}), and an RL version (called \emph{anytime PSRO (APSRO)}).

\emph{Anytime double oracle (ADO)} is a modification of \emph{range of skill (ROS)}~\cite{zinkevich07ros} that finds a restricted Nash equilibrium over two restricted games: one per player. Each player's restricted game is defined such that their strategies are restricted to be within their population, but the opponent is unrestricted. ADO adds a best response to each player's restricted distribution to the other player's population. ADO is guaranteed to not increase exploitability from one iteration to the next and is guaranteed to converge to a Nash equilibrium in a number of iterations at most equal to the number of pure strategies in the game.  

Next, we show that an existing method, \emph{regret minimization against a best response (RM-BR)} \cite{johanson12cfrbr}, can be used to calculate the least exploitable restricted distribution in a double oracle algorithm via regret minimization against a best response. This algorithm, which we refer to as RM-BR DO, can efficiently compute a least-exploitable restricted distribution because without solving the large restricted game. However, in large games where computing best responses via RL is necessary, both of these methods becomes prohibitively expensive.

\emph{Anytime Policy Space Response Oracles (APSRO)} solves this problem by updating the restricted distribution using a no-regret algorithm trained against a single best response that is itself being continually trained via reinforcement learning against the restricted distribution. Although this method lacks theoretical guarantees enjoyed by the exact methods, we find that in practice it can perform much better than PSRO, and tends to not increase exploitability. To the best of our knowledge, ADO is the first double oracle method that both (1) is proven to converge to a Nash equilibrium in a number of iterations bounded by the number of pure strategies of the game; and (2) does not increase exploitability from one iteration to the next. APSRO scales up ADO to large games using reinforcement learning. 

To summarize, our contributions are as follows:
\begin{itemize}
    \item We present ADO, a double oracle algorithm that is guaranteed to not increase exploitability from one iteration to the next. 
    \item We present RM-BR DO, a method for efficiently finding the least-exploitable restricted distribution in ADO via a no-regret procedure against best responses. 
    \item We present APSRO, a reinforcement learning version of ADO that outperforms PSRO in all of our experiments and, unlike PSRO, tends to not increase exploitability.
\end{itemize}

\section{Background}
We consider extensive-form games with perfect recall \citep{hansen2004dynamic}. An extensive-form game progresses through a sequence of player actions, and has a \emph{world state} $w \in \mathcal{W}$ at each step. 
In an $N$-player game, $\mathcal{A} = \mathcal{A}_1 \times \cdots \times \mathcal{A}_N$ is the space of joint actions for the players. $\mathcal{A}_i(w) \subseteq \mathcal{A}_i$ denotes the set of legal actions for player $i \in \mathcal{N} = \{1, \ldots, N\}$ at world state $w$ and $a = (a_1, \ldots, a_N) \in \mathcal{A}$ denotes a joint action. At each world state, after the players choose a joint action, a transition function $\mathcal{T}(w, a) \in \Delta^\mathcal{W}$ determines the probability distribution of the next world state $w'$. Upon transition from world state $w$ to $w'$ via joint action $a$, player $i$ makes an \emph{observation} $o_i = \mathcal{O}_i(w,a,w')$. In each world state $w$, player $i$ receives a reward $\mathcal{R}_i(w)$. The game ends when the players reach a terminal world state. In this paper, we consider games that are guaranteed to end in a finite number of actions.

A \emph{history} is a sequence of actions and world states, denoted $h = (w^0, a^0, w^1, a^1, \ldots, w^t)$, where $w^0$ is the known initial world state of the game. $\mathcal{R}_i(h)$ and $\mathcal{A}_i(h)$ are, respectively, the reward and set of legal actions for player $i$ in the last world state of a history $h$. An \emph{information set} for player $i$, denoted by $s_i$, is a sequence of that player's observations and actions up until that time $s_i(h) = (a_i^0, o_i^1, a_i^1, \ldots, o_i^t)$. Define the set of all information sets for player $i$ to be $\mathcal{I}_i$. 
The set of histories that correspond to an information set $s_i$ is denoted $\mathcal{H}(s_i) = \{ h: s_i(h) = s_i \}$, and it is assumed that they all share the same set of legal actions $\mathcal{A}_i(s_i(h)) = \mathcal{A}_i(h)$. 

A player's \emph{strategy} $\pi_i$ 
is a function mapping from an information set to a probability distribution over actions. A \emph{strategy profile} $\pi$ is a tuple $(\pi_1, \ldots, \pi_N)$. All players other than $i$ are denoted $-i$, and their strategies are jointly denoted $\pi_{-i}$. A strategy for a history $h$ is denoted $\pi_i(h) = \pi_i(s_i(h))$ and $\pi(h)$ is the corresponding strategy profile. 
When a strategy $\pi_i$ is learned through RL, we refer to the learned strategy as a \emph{policy}.

The \emph{expected value (EV)} $v_i^{\pi}(h)$ for player $i$ is the expected sum of future rewards for player $i$ in history $h$, when all players play strategy profile $\pi$. The EV for an information set $s_i$ is denoted $v_i^{\pi}(s_i)$ and the EV for the entire game is denoted $v_i(\pi)$. A \emph{two-player zero-sum} game has $v_1(\pi) + v_2(\pi) = 0$ for all strategy profiles $\pi$. The EV for an action in an information set is denoted $v_i^{\pi}(s_i,a_i)$. A \emph{Nash equilibrium (NE)} is a strategy profile such that, if all players played their NE strategy, no player could achieve higher EV by deviating from it. Formally, $\pi^*$ is a NE if $v_i(\pi^*) = \max_{\pi_i}v_i(\pi_i, \pi^*_{-i})$ for each player $i$.

The \emph{exploitability} $e(\pi)$ of a strategy profile $\pi$ is defined as $e(\pi) = \sum_{i \in \mathcal{N}} \max_{\pi'_i}v_i(\pi'_i, \pi_{-i})$. A \emph{best response (BR)} strategy $\mathbb{BR}_i(\pi_{-i})$ for player $i$ to a strategy $\pi_{-i}$ is a strategy that maximally exploits $\pi_{-i}$: $\mathbb{BR}_i(\pi_{-i}) = \arg\max_{\pi_i}v_i(\pi_i, \pi_{-i})$. An \emph{$\bm{\epsilon}$-best response ($\bm{\epsilon}$-BR)} strategy $\mathbb{BR}^\epsilon_i(\pi_{-i})$ for player $i$ to a strategy $\pi_{-i}$ is a strategy that is at most $\epsilon$ worse for player $i$ than the best response: $v_i(\mathbb{BR}^\epsilon_i(\pi_{-i}), \pi_{-i}) \ge v_i(\mathbb{BR}_i(\pi_{-i}), \pi_{-i}) - \epsilon$. An \emph{$\bm{\epsilon}$-Nash equilibrium ($\bm{\epsilon}$-NE)} is a strategy profile $\pi$ in which, for each player $i$, $\pi_i$ is an $\epsilon$-BR to $\pi_{-i}$. 

A \emph{normal-form game} is a single-step extensive-form game. An extensive-form game induces a normal-form game in which the legal actions for player $i$ are its deterministic strategies $\bigtimes_{s_i \in \mathcal{I}_i} \mathcal{A}_i(s_i)$. These deterministic strategies are called \emph{pure strategies} of the normal-form game. A \emph{mixed strategy} is a distribution over a player's pure strategies. 

\section{Related Work}


\subsection{Double Oracle (DO) and Policy Space Response Oracles (PSRO)}
Double Oracle~\citep{double_oracle} is an algorithm for finding a NE in normal-form games. The algorithm works by keeping a population of strategies $\Pi^t$ at time $t$. Each iteration a NE $\pi^{*,t}$ is computed for the game restricted to strategies in $\Pi^t$. Then, a best response to this NE for each player $\mathbb{BR}_i(\pi^{*,t}_{-i})$ is computed and added to the population $\Pi_i^{t+1} = \Pi_i^t \cup \{\mathbb{BR}_i(\pi^{*,t}_{-i}) \}$ for $i \in \{1, 2\}$. The DO algorithm is described in Algorithm \ref{double_oracle}. Although in the worst case DO must expand all pure strategies, in many games DO empirically terminates early and outperforms existing methods. An interesting open problem is characterizing games where DO will outperform existing methods.

\begin{algorithm}[tb]
   \caption{Double Oracle \cite{double_oracle}}
   \label{double_oracle}
\begin{algorithmic}
 \STATE {\bfseries Result:} Nash Equilibrium
   \STATE {\bfseries Input:} Initial population $\Pi^0$
   \REPEAT[for $t=0,1,\ldots$]
   \STATE $\pi^r \gets$ NE in game restricted to strategies in $\Pi^t$
  \FOR{$i \in \{1,2\}$}
   \STATE Find a best response $\beta_i \gets \mathbb{BR}_i(\pi^r_{-i})$
   \STATE $\Pi^{t+1}_i \gets \Pi^t_i \cup \{\beta_i\}$
  \ENDFOR
  \UNTIL{No novel best response exists for either player}
   \STATE {\bfseries Return:} $\pi^r$
\end{algorithmic}
\end{algorithm}


Policy-Space Response Oracles (PSRO) \cite{psro} approximates the Double Oracle algorithm. The restricted-game NE is computed on the empirical game matrix $U^\Pi$, generated by having each policy in the population $\Pi$ play each opponent policy and tracking average utility in a $\Pi_1 \times \Pi_2$ payoff matrix \citep{wellman2006methods}. In each iteration, an approximate best response to the current restricted NE over the policies is computed via any RL algorithm. 


\subsection{Minimum-Regret Constrained Profile}
The concept of finding a low-exploitability distribution in a restricted game was also explored in \citet{jordan2010strategy} and \citet{wang2021evaluating}, where they define the \emph{minimum regret constrained profile} as the distribution over a restricted population that achieves the lowest exploitability. In this paper, we use the term \emph{least-exploitable restricted distribution} instead of minimum regret constrained profile. We demonstrate that RM-BR \cite{johanson12cfrbr} can be used to calculate the least-exploitable restricted distribution, and we incorporate this method into a DO algorithm that converges to an approximate Nash equilibrium through our algorithm RM-BR DO. We also present APSRO, which extends RM-BR DO to large games through RL. 

\subsection{Range of Skill (ROS)}
This paper presents ADO, a modification to the \emph{range of skill (ROS)} algorithm introduced in \cite{zinkevich07ros} and further explored in \cite{hansen08onros}. 
ROS is a variant of the DO algorithm that produces a series of restricted games by iteratively adding new strategies. 
As in ADO, ROS defines two seperate restricted games where one player is restricted to play strategies in their population while the other player is unrestricted. Also as in ADO, the restricted distribution is part of a Nash equilibrium strategy profile for the restricted game and is therefore the least-exploitable restricted distribution. However, ADO and ROS differ in the strategy that they add to the unrestricted player’s population.  ADO adds a best response to a restricted Nash equilibrium, while ROS adds a strategy that is part of a restricted NE.

This difference proves crucial when scaling up to large games. In large games, solving the restricted game where one player is unrestricted is infeasible and as a result methods based on ROS cannot scale to large games. Alternatively, since ADO only adds best responses, it naturally scales to large games via APSRO where the best responses are learned through RL. Additionally, while ROS, like ADO, decreases exploitability monotonically and performs well in practice, the only known convergence guarantees for ROS are asymptotic with a convergence rate that is exponential in the size of the game \cite{hansen08onros}. Conversely, ADO is guaranteed to converge in a number of iterations at most the number of pure strategies in a game. Finally, ROS is only guaranteed to reach an $\epsilon$-Nash, not an exact NE, whereas ADO reaches an exact NE upon termination.  

\begin{figure*}[ht]
    \centering
    \includegraphics[width=\textwidth]{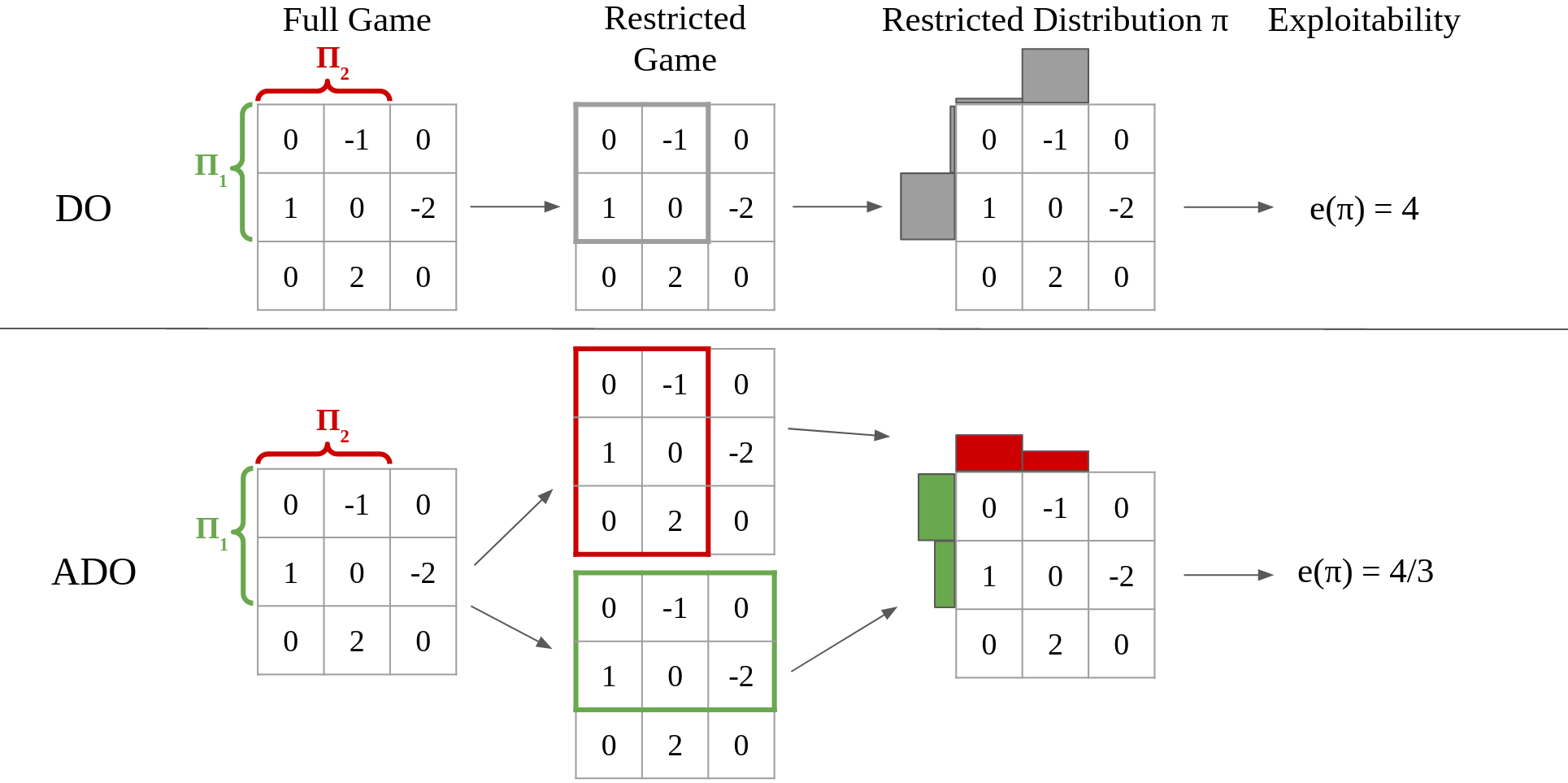}
    \caption{In DO, a single restricted game is created and solved. Since this restricted game does not consider strategies outside of the population, it can lead to exploitable restricted distributions, as shown in the top figure. Conversely, ADO creates two restricted games where the opponent is unrestricted. These restricted games are then solved which results in the least-exploitable restricted distributions for both players.}
    \label{fig:ADO_Diagram}
\end{figure*}

\section{Anytime Double Oracle (ADO)}

The double oracle (DO) algorithm is guaranteed to converge because in the worst case it must expand all pure strategies, at which point it terminates at a Nash equilibrium (NE). Unfortunately, before convergence, there is no guarantee on the exploitability of the restricted-game NE. In fact, DO can increase exploitability from one iteration to the next. 

To see this, consider the game in Figure \ref{fig:ADO_Diagram}. If both players start with a population consisting only of the first strategy (top row and left column), then the best response for each player is the second strategy, giving that player value 1 
, for a total exploitability of 2. However, in the next iteration DO will include both the first and second strategies in the population for both players, and the restricted-game NE will give probability 1 to the second strategy. This restricted NE has exploitability of 4. We also provide a generalization to this bad case in Section \ref{sec:exp} and show empirically that DO does indeed arbitrarily increase exploitability before terminating in this class of games.   

In this paper, we introduce Anytime Double Oracle (ADO), which is guaranteed to not increase exploitability from one iteration to the next. Like DO, ADO maintains a population $\Pi^t$ at every timestep $t$. Also like DO, in every iteration a Nash equilibrium is computed on a restricted game and a best response to this restricted NE is computed and added to both populations. However, unlike DO, a different restricted game is created for each player. The restricted game $G^i$ for player $i$ is created by restricting that player to only play strategies included in their population $\Pi_i$, while the opponent can play any strategy in the full game. The game value of $G^i$ for player $i$ is 
\begin{equation}
\label{restricted_game}
\max_{\pi_i \in \Pi_i}\min_{\pi_{-i}}v_i(\pi_i, \pi_{-i}).   
\end{equation}
The restricted game $G^i$ for player $i$ is then solved for both players to get a restricted NE $(\pi^i_1, \pi^i_2)$ for both players. We refer to the restricted player's mixed strategy as the restricted NE $\pi^r_i = \pi^i_i$. The restricted NE for player $i$ is the least exploitable mixed strategy supported by player $i$'s population. Note that in large games this restricted game will be prohibitively large to solve, and will require APSRO, introduced later in this paper.

Next, a best response $\beta_i = \mathbb{BR}(\pi^r_{-i})$ is computed for each player $i$ against the restricted-NE mixed strategy of the restricted opponent, and is added to the player's population. If there are multiple best responses, a novel best response $\beta_i \not\in \Pi_i$ is chosen that is not currently in that player's population. ADO is described in Algorithm \ref{ADO}.

\begin{algorithm}
 \caption{Anytime Double Oracle (ADO)}
 \label{ADO}
 \begin{algorithmic}
 \STATE {\bfseries Result:} Nash Equilibrium
 \STATE {\bfseries Input:} Initial population $\Pi^0$
 \REPEAT[for $t=0,1,\ldots$]
 \STATE $\pi^r_i \gets$ NE in restricted game $G^i$ (eq. (\ref{restricted_game})), for $i \in \{1, 2\}$
 \FOR{$i \in \{1,2\}$}
 \STATE Find a novel best response $\beta_i \gets \mathbb{BR}_i(\pi^r_{-i})$
 \STATE $\Pi^{t+1}_i = \Pi^t_i \cup \{ \beta_i \}$
 \ENDFOR 
 \UNTIL{No novel best response exists for either player}
 \STATE {\bfseries Return:} $\pi^r$
 \end{algorithmic}
\end{algorithm}

ADO is guaranteed to terminate because there are finitely many pure strategies in the original game. When ADO terminates, the restricted NE is a Nash equilibrium in the original game (Proposition \ref{prop:converge}). Unlike DO, the exploitability of the restricted NE does not increases from one iteration to the next (Proposition \ref{prop:monotonic}).

\begin{figure*}[ht]
    \centering
    \includegraphics[width=\textwidth]{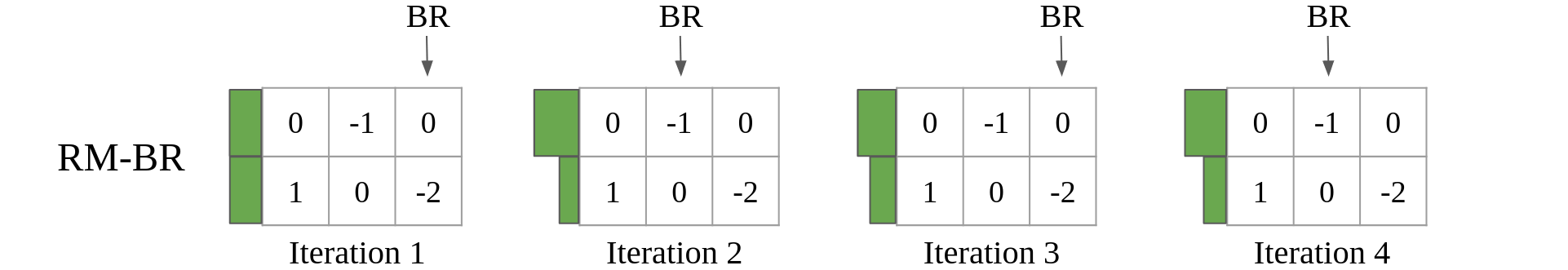}
    \caption{RM-BR converges to the Nash equilibrium of the restricted game by updating the restricted distribution via a no-regret algorithm against a best response at every iteration. Here the best response to the current restricted distribution is indicated by the arrow and the restricted distribution is in green. RM-BR allows us to produce a low-exploitability restricted distribution without solving the full restricted game. This basic idea is the core of APSRO as well, where we use the best response that is being trained as a proxy for the true best response and update the restricted distribution the same way.}
    \label{fig:RMBR_Diagram}
\end{figure*}

\begin{proposition}\label{prop:monotonic}
The exploitability of ADO is monotonically non-increasing. 
\end{proposition}
\begin{proof}
Let $\pi^t$ be the restricted NE in ADO at iteration $t$. Then for player $i$, since $\Pi^t_i \subseteq \Pi^{t+1}_i$
\begin{equation}
\begin{aligned}
    v_i(\pi^t_i, \mathbb{BR}_{-i}(\pi^t_i)) &= \max_{\pi_i \in \Pi_i^t}\min_{\pi_{-i}}v_i(\pi_i, \pi_{-i}) \\ 
    &\leq \max_{\pi_i \in \Pi^{t+1}_i}\min_{\pi_{-i}}v_i(\pi_i, \pi_{-i}) \\ 
    &= v_i(\pi^{t+1}_i, \mathbb{BR}_{-i}(\pi^{t+1}_i)).
\end{aligned}
\end{equation}
Since each player's value is monotonically non-decreasing, the value of the best response is non-increasing, and the exploitability of $\pi^t$ is also non-increasing:
\begin{equation}
\begin{aligned}
    e(\pi^{t+1}) :={}& -\sum_{i} v_i(\pi^{t+1}_i, \mathbb{BR}_{-i}(\pi^{t+1}_i)) \\
    \le{}& -\sum_{i} v_i(\pi^t_i, \mathbb{BR}_{-i}(\pi^t_i)) = e(\pi^t).
\end{aligned}
\end{equation}
\end{proof}

To illustrate this property of ADO, consider the algorithm dynamics on the DO bad case given in Figure \ref{fig:ADO_Diagram}. Similar to DO, ADO adds the second strategy to the population in the first iteration. Now, however, instead of taking the second strategy with probability 1 as DO does, ADO solves the restricted game where one player is restricted to the first two strategies and the other is unrestricted and can play all three strategies. The Nash equilibrium of this game for the restricted player is to play the first strategy with probability $\frac{2}{3}$ and to play the second strategy with probability $\frac{1}{3}$. This strategy results in a total exploitability of $\frac{4}{3}$, compared with the DO exploitability of $4$ and the initial ADO exploitability of $2$. In addition to this property of never increasing exploitability, ADO is guaranteed to converge to a Nash equilibrium, as shown below. 

\begin{proposition}\label{prop:converge}
When ADO terminates, the restricted NE of both players is a Nash equilibrium in the full game.  
\end{proposition}
\begin{proof}
Let $(\pi^r_1, \pi'_2)$ and $(\pi'_1, \pi^r_2)$ be the Nash equilibria in the restricted games $G^1$ and $G^2$ for player 1 and 2, respectively, upon termination. If $\pi_1'$ or $\pi_2'$ have support outside the population, ADO would not terminate, because there would exist another novel best response. So the support of both $\pi_1'$ and $\pi_2'$ must be inside the population, which makes them feasible for their respective player's restricted game. Then 
\begin{equation}
    \begin{aligned}
        v_1(\pi^r_1, \pi^r_2) &\le v_1(\pi_1', \pi^r_2) \\
        &\leq v_1(\pi_1', \pi_2') \\ 
        &\leq v_1(\pi^r_1, \pi_2') \\
        &\leq v_1(\pi^r_1, \pi^r_2).
    \end{aligned}
\end{equation}
The four inequalities follow, in order, because: (a) player 1 doesn't want to deviate from $\pi_1'$ to $\pi^r_1$ in $G^2$; (b) player 2 doesn't want to deviate from $\pi^r_2$ to $\pi_2'$ in $G^2$; (c) player 1 doesn't want to deviate from $\pi^r_1$ to $\pi_1'$ in $G^1$; and (d) player 2 doesn't want to deviate from $\pi_2'$ to $\pi^r_2$ in $G^1$.

Therefore, $v_1(\pi^r_1, \pi^r_2) = v_1(\pi_1', \pi^r_2)$ which implies that player 1 has no incentive to deviate from $\pi^r_1$ to $\pi_1'$ or any other strategy against $\pi^r_2$. A symmetric argument holds for player 2, implying that $\pi^r$ is a Nash equilibrium in the full game.
\end{proof}

\begin{figure*}[ht]
    \centering
    \includegraphics[width=\textwidth]{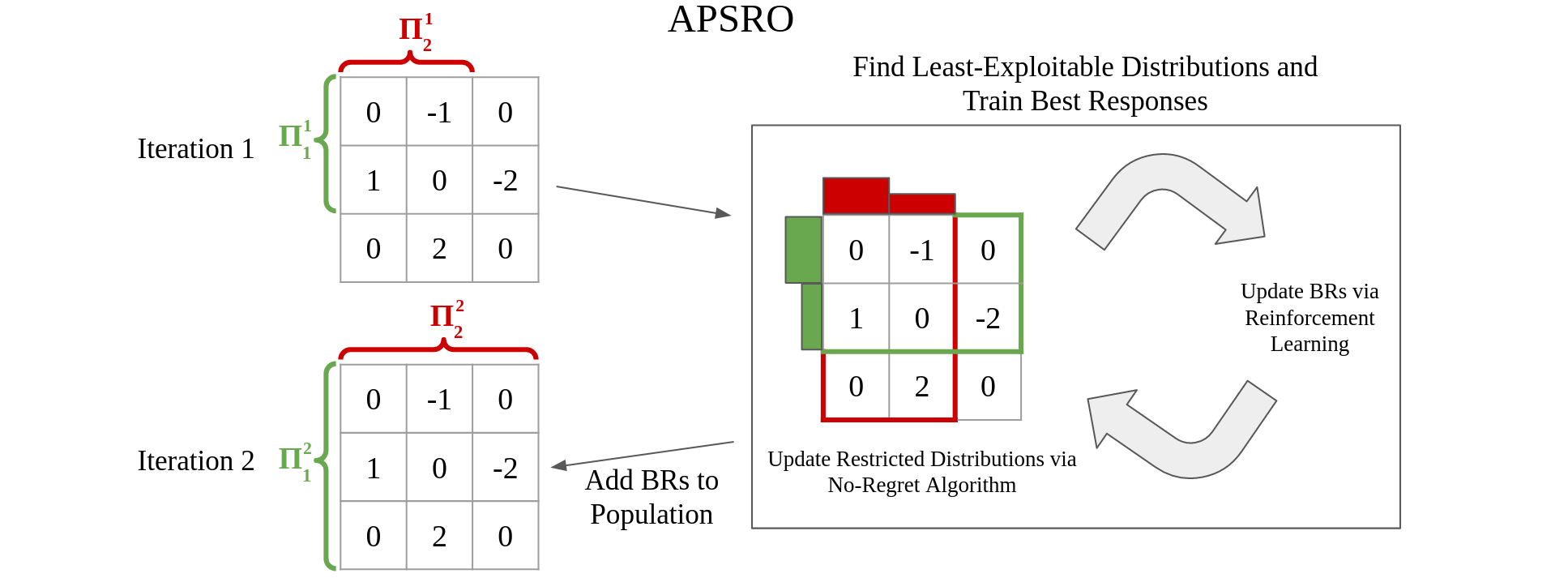}
    \caption{In APSRO, the following happens in one iteration: (1) Two restricted games are created where one player is unrestricted, one for each player. (2) For both players, a BR is trained against the restricted distribution while the restricted distribution is updated via a no-regret algorithm against this BR. (3) This BR is then added to the population.}
    \label{fig:APSRO_Diagram}
\end{figure*}

\section{Regret-Minimizing against a BR Double Oracle (RM-BR DO)}

In this section we propose a specific form of ADO called RM-BR DO, which finds  least exploitable restricted strategies $\pi^r_i$ for the restricted player of each restricted game. RM-BR DO uses a regret minimization vs. a best response (RM-BR) method that is similar to CFR-BR~\cite{johanson12cfrbr}: the restricted player uses regret minimization (RM) to find its mixed strategy, alternating with the unrestricted opponent finding best responses (BR). Interestingly, despite not computing the full restricted NE in each restricted game, RM-BR DO is guaranteed to not increase exploitability from one iteration to the next by more than the approximation errors of the RM and BR oracles. 


For a given game (later, this will be each restricted game), RM-BR \citep{johanson12cfrbr} works as follows. First, a regret-minimizing (RM) algorithm, such as regret matching \cite{hart2000simple}, is initialized. Then, in every iteration, a best response to the current RM strategy is calculated. Next, the RM strategy is updated with respect to this best response. In this work, we use the Exp3 update rule \cite{auer1995gambling} to update the RM strategy, but any regret-minimizing algorithm would work. RM-BR is described in Algorithm \ref{RM-BR}.

\begin{algorithm}
  \caption{RM-BR}
  \label{RM-BR}
  \begin{algorithmic}
  \STATE {\bfseries Result:} Approximately Least-Exploitable Restricted NE
   \STATE {\bfseries Input:} population $\Pi_i$
   \STATE Initialize $\pi^0$
   \FOR{$t=0,...,n-1$}
   \FOR{$i \in \{1,2\}$}
   \STATE Find $\beta_{-i} \gets \mathbb{BR}_{-i}(\pi^t_{i})$
   \STATE Update $\pi^{t+1}_i$ from $\pi^t_i$ via RM vs. $\beta_{-i}$
   \ENDFOR
   \ENDFOR 
 \STATE {\bfseries Return:} $\pi^n$
  \end{algorithmic}
\end{algorithm}

\begin{proposition}
RM-BR with a regret minimizing algorithm that has regret $R_t$ at iteration $t$ will output a distribution $\pi^n$ such that
$e(\pi^n)
\leq \frac{R_n}{n}$.
\end{proposition}
\begin{proof}
The proof follows the same argument as used in Theorem 3 of \citet{johanson12cfrbr}.
\end{proof}


In RM-BR DO, we use RM-BR to compute the restricted NE $\pi^r$ in each iteration of ADO. Formally, RM-BR DO keeps a population of pure strategies $\Pi^t$ at time $t$. In each iteration, $\pi^r$ is initialized to a uniform distribution over pure strategies in both populations, and in an inner loop updated against a series of best responses, as in RM-BR.
The final best response to the restricted NE is added as a strategy to the population for both players. 
The algorithm terminates when the difference between the game value of each player's restricted game is less than $\epsilon$
, indicating that $\pi^r$ is an $\epsilon$-NE in the original game~\citep{hansen08onros}.

\begin{algorithm}
 \caption{RM-BR DO}
 \begin{algorithmic}
 \STATE {\bfseries Result:} Approximate Nash Equilibrium
 \STATE {\bfseries Input:} Initial population $\Pi^0$
 \REPEAT[for $t=0,1,\ldots$]
 \STATE Compute restricted NE $\pi^r$\ via RM-BR\;
  \FOR{$i \in \{1,2\}$}
 \STATE Find $\beta_i \gets \mathbb{BR}_i(\pi^r_{-i})$
 \STATE $\Pi^{t+1}_i = \Pi^t_i \cup \{ \beta_i \}$
 \ENDFOR
 \UNTIL{No novel best response exists for either player}
 \STATE {\bfseries Return:} $\pi^r$
 \end{algorithmic}
\end{algorithm}

\begin{proposition}
Assume RM-BR DO runs sufficiently many RM-BR inner-loop updates in each iteration such that the exploitability in each restricted game is at most $\epsilon$. Then the exploitability of RM-BR DO will never increase by more than $2\epsilon$ from one iteration to the next. 
\end{proposition}
\begin{proof}
Let $\pi^t$ be the restricted NE of RM-BR DO at iteration $t$. 
Then since $\Pi^t_i \subseteq \Pi^{t+1}_i$
\begin{equation}
\begin{aligned}
    v_i(\pi^t_i, \mathbb{BR}_{-i}(\pi^t_i)) &\leq \max_{\pi_i \in \Pi_i^t}\min_{\pi_{-i}}v_i(\pi_i, \pi_{-i}) \\ 
    &\leq \max_{\pi_i \in \Pi^{t+1}_i}\min_{\pi_{-i}}v_i(\pi_i, \pi_{-i}) \\ 
    &\leq v_i(\pi^{t+1}_i, \mathbb{BR}_{-i}(\pi^{t+1}_i)) + \epsilon.
\end{aligned}
\end{equation}
\end{proof}


\section{Anytime PSRO (APSRO)}
Although the exploitability of RM-BR DO will never increase much, it needs to compute many different best responses in each inner loop iteration. This can quickly become infeasible, especially in large games. In this section we introduce a scalable version of ADO, called Anytime PSRO (APSRO). 

\begin{algorithm}
\caption{Anytime PSRO}
\begin{algorithmic}
 \STATE {\bfseries Result:} Approximate Nash Equilibrium
 \STATE {\bfseries Input:} Initial population $\Pi^0$
 \WHILE{Not Terminated \{$t=0,1,\ldots$\}}
 \STATE Initialize $\pi^r_i$ to uniform over $\Pi^t_i$ for $i \in \{1, 2\}$\;
 \FOR{$i \in \{1,2\}$} 
 \FOR{$n$ iterations}
 \FOR{$m$ iterations}
 \STATE Update policy $\beta_{-i}$ toward $\mathbb{BR}_{-i}(\pi^r_{i})$
 \ENDFOR
 \STATE Update $\pi^r_i$ via regret minimization vs. $\beta_{-i}$\;
 \ENDFOR
 \ENDFOR
 \STATE $\Pi^{t+1}_i = \Pi^t_i \cup \{\beta_i\}$ for $i \in \{1, 2\}$\;
 \ENDWHILE 
 \STATE {\bfseries Return:} $\pi^r$
\end{algorithmic}
\end{algorithm}


\begin{figure*}[t]
    \centering
    \begin{subfigure}[b]{0.33\textwidth}
        \centering
        \includegraphics[width=\textwidth]{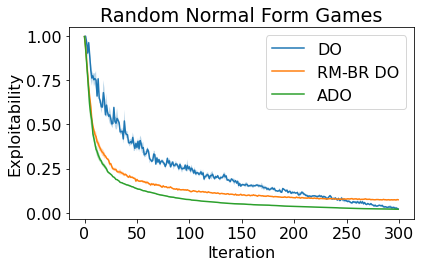}
        \caption{Random Normal Form Games}
        \label{fig:random}
    \end{subfigure}
    \begin{subfigure}[b]{0.33\textwidth}
        \centering
        \includegraphics[width=\textwidth]{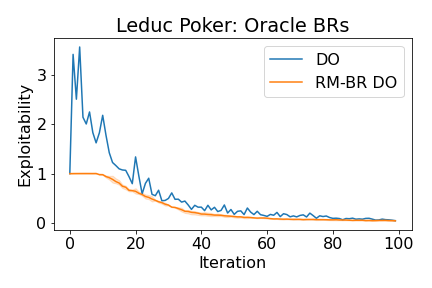}
        \caption{Leduc with Oracle Best Responses}
        \label{fig:leduc_oracle}
    \end{subfigure}
    \begin{subfigure}[b]{0.33\textwidth}
        \centering
        \includegraphics[width=\textwidth]{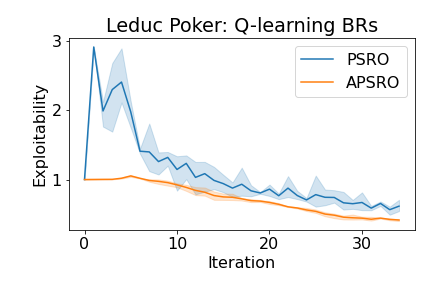}
        \caption{Leduc with Q-Learning Best Responses}
        \label{fig:leduc_q}
    \end{subfigure}
    \caption{Tabular Methods}
\end{figure*}

Instead of recomputing a best response in every inner-loop iteration, APSRO maintains one approximate best response RL policy and updates it for a small number of steps in each inner-loop iteration. 
To speed up the algorithm, APSRO can update the approximate best response in fewer iterations than are necessary to get a full best response. This hyperparameter for the number of best-response updates per every Exp3 update trades off between theoretical guarantees and speed. 

The updates to the best response can be made through a variety of algorithms. In this paper we show experiments with updates via tabular Q-learning as well as deep reinforcement learning.



\section{Experiments}\label{sec:exp}
We report results on normal form games, Leduc poker, Goofspiel, and a continuous-action game introduced in \citet{mcaleer2021xdo}. Details of these games are included in the appendix. 

\subsection{Experiments with Tabular Methods}
In this section we report ADO and RM-BR DO results with oracle best responses and tabular Q learning best responses. In all ADO experiments we solve the restricted games using a linear program.   

\subsubsection{DO Bad Case}

To demonstrate how DO can increase exploitability every iteration besides the last, we present a generalization of the game presented in Figure \ref{fig:ADO_Diagram}. Consider a game where all values are $0$, except if the row index $r$ is one more than the column index $c$, in which case the value is $\sum_{i=0}^{r} 2^i + 2i$, or if the column index $c$ is one more than the row index $r$, 
in which case the value is $\sum_{i=0}^c -2^i + 2i$. 
\begin{wrapfigure}{r}{0.23\textwidth}
  \begin{center}
    \includegraphics[width=0.23\textwidth]{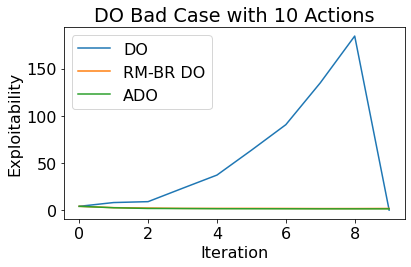}
  \end{center}
    \caption{DO Bad Case}
    \label{fig:DO_Counterexample}
\end{wrapfigure}
We plot the performance of DO and ADO in this game with 9 actions in Figure \ref{fig:DO_Counterexample} and show that DO increases exploitability in every iteration except the last, while ADO does not increase exploitability.

\subsubsection{Random Normal-Form Games}
To create random normal-form games, we sample values from Uniform(0,1). Figure \ref{fig:random} shows results on random normal-form games with 500 actions. We see that ADO greatly outperforms DO and tends to not increase exploitability. 

\begin{figure*}[t]
    \centering
    \begin{subfigure}[b]{0.33\textwidth}
        \centering
        \includegraphics[width=\textwidth]{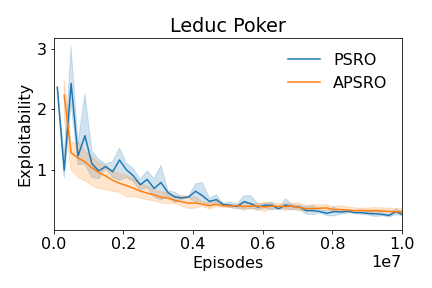}
        \caption{Leduc with DDQN BRs}
        \label{fig:leduc_APSRO}
    \end{subfigure}
    \begin{subfigure}[b]{0.33\textwidth}
        \centering
        \includegraphics[width=\textwidth]{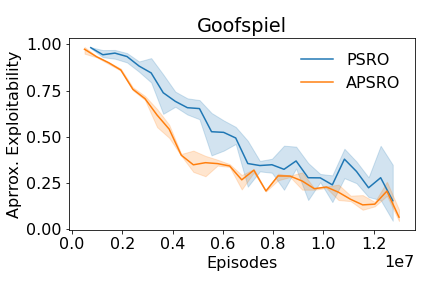}
        \caption{Goofspiel with DDQN BRs}
        \label{fig:goofspiel}
    \end{subfigure}
    \begin{subfigure}[b]{0.33\textwidth}
        \centering
        \includegraphics[width=\textwidth]{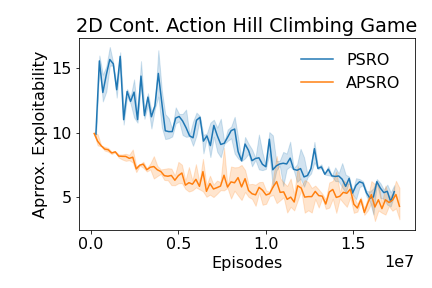}
        \caption{Continuous-Action Hill-Climbing Game}
        \label{fig:hill}
    \end{subfigure}
    \caption{Deep RL Methods}
\end{figure*}

\subsubsection{Leduc Poker}
We test RM-BR DO on Leduc poker in OpenSpiel~\cite{lanctot2019openspiel} and compare it with DO. As shown in Figure \ref{fig:leduc_oracle}, RM-BR DO outperforms DO and does not increase exploitability, but the performance of the two algorithms converges at the end. 
Although RM-BR DO requires many inner-loop best-response iterations, it is guaranteed to not increase exploitability much from one iteration to the next, and APSRO approximates this process.

\subsubsection{Leduc Poker with Tabular Q Learning Best Responses}
In this section we have a tabular Q learning agent as the best response. Similar to the ADO results, in Figure \ref{fig:leduc_q} we see that APSRO is able to perform better than PSRO and does not increase exploitability. However, as in PSRO, the ending exploitability is still much higher than the exploitability that a method like CFR \cite{zinkevich2008regret, brown2019solving, farina2019stable} would achieve in Leduc.

\subsection{Experiments with Deep Reinforcement Learning}
In this section we use deep reinforcement learning for the best response for both APSRO and PSRO, and we present results on Leduc poker, Goofspiel, and a continuous-action hill climbing game. Training details are described in the appendix. To calculate approximate exploitability in Goofspiel and the continuous-action hill climbing game, we train reinforcement learning best responses from scratch against checkpoints. 

\subsubsection{Leduc Poker}
As shown in Figure \ref{fig:leduc_APSRO}, APSRO with deep RL best responses performs competitively with PSRO but has a much smoother exploitability curve. 
PSRO with deep RL best responses exhibits a similar fluctuations in exploitability as previously demonstrated with PSRO with tabular Q learning best responses. However, APSRO starts out at a higher exploitability than with tabular Q learning. We conjecture that this is likely due to different initialization between tabular and neural best-response methods. We use DDQN \cite{van2016deep} for the deep RL best responses.

\subsubsection{Goofspiel}
We experiment on the standard Goofspiel game in Openspiel \cite{lanctot2019openspiel} with 13 cards and simultaneous moves.  As shown in figure \ref{fig:goofspiel}, APSRO with deep RL best responses outperforms PSRO and has a slightly smoother exploitability curve. 
Goofspiel is a much larger game than Leduc poker, suggesting that APSRO can scale to large games. We use DDQN \cite{van2016deep} for the deep RL best response. 

\subsubsection{Continuous-Action Hill-Climbing Game}
The 2-dimensional continuous-action hill-climbing game was introduced in \cite{mcaleer2021xdo} as a simple continuous-action game where PSRO performs relatively poorly. As shown in figure \ref{fig:hill}, PSRO indeed increases exploitability many times and does not reach the same exploitability as a random initialization until roughly 5,000,000 episodes. Conversely, APSRO rarely increases exploitability and stays below PSRO throughout the entire experiment. We use PPO \cite{schulman2017proximal} for the deep RL best response.

\section{Discussion}
We introduced ADO, a modification of DO that does not increase exploitability much from one iteration to the next. We also introduced RM-BR DO, a modification of ADO that finds the least exploitable restricted distribution via a no-regret procedure. Finally, we introduced APSRO, a procedure that approximates ADO but can handle reinforcement learning best responses. As shown in our experiments, ADO and RM-BR PSRO outperform DO on normal form games and APSRO outperforms PSRO on Leduc poker. Existing approaches to improving PSRO could potentially be combined with our method \citep{mcaleer2021xdo, balduzzi2019open, smith2021iterative, feng2021discovering, mcaleer2020pipeline, wright2019iterated}.

Regarding worst-case convergence time, it is known that some games could have long sequences of subgames, even when the unrestricted subgames are fully solved~\citep{hansen08onros}. This issue is a potential problem for both ADO and DO.
How one could reduce the size of these worst-case sequences is still an open question, and an interesting direction for future work. 

One downside of ADO compared to DO is that it requires best responses to be novel best responses, and can get stuck at high exploitability if best responses are not chosen to be novel. This could pose a problem in APSRO, where finding novel best responses could be difficult. One possible method to encourage novelty would be to include a diversity reward for the BR. 


\bibliography{main.bib}
\bibliographystyle{icml2021}


\newpage
\appendix
\onecolumn

\section{No-Regret Algorithms}
\subsection{Exp3}
The exponential-weight algorithm for exploration and exploitation (Exp3) \cite{auer1995gambling} is an adversarial bandit method that has sublinear regret. We used the EXP3.P variant from \cite{bubeck2012regret}.

\begin{algorithm}[h]
\caption{Exp3}
\label{Exp3}
\begin{algorithmic}
 \STATE {\bfseries Input:} $n$ iterations, $k$ actions, learning rate $\eta$
 \STATE Initialize cumulative rewards $\hat{S_0} = (0,0,...0)$
 \FOR{$t=1,...,n$}
 \STATE Calculate the sampling distribution $P_{t,i}$: $P_{t,i} = (1-\gamma)\frac{exp(\eta\hat{S}_{t-1,i})}{\sum^k_{j=1}exp(\eta\hat{S}_{t-1,j})} + \frac{\gamma}{k}$ for each $i \in [1..k]$
 \STATE Sample action $A_t \sim P_t$ and observe reward $X_t$
 \STATE Calculate $\hat{S}_{t,i}$: $\hat{S}_{t,i} = \hat{S}_{t-1,i} + \frac{X_t \mathbbm{1}\{A_t=i\}}{P_{t,i}}$
 \ENDFOR
\end{algorithmic}
\end{algorithm}

\subsection{Multiplicative Weights Update}
The Multiplicative Weights Update (MWU) algorithm is an online learning method which converges in time-average to Nash equilibrium \cite{cesa2006prediction, freund1999adaptive}.  

\begin{algorithm}[h]
\caption{Multiplicative Weights Update}
\label{MWU}
\begin{algorithmic}
 \STATE {\bfseries Input:} $n$ iterations, $k$ actions, learning rate $\eta$
 \STATE Initialize cumulative rewards $\hat{S_0} = (0,0,...0)$
 \FOR{$t=1,...,n$}
 \STATE Calculate the sampling distribution $P_{t,i}$: $P_{t,i} = \frac{exp(\eta\hat{S}_{t-1,i})}{\sum^k_{j=1}exp(\eta\hat{S}_{t-1,j})}$ for each $i \in [1..k]$
 \STATE Observe reward $X_{t,i}$ for each action $i \in [1..k]$
 \FOR{$i=1,...k$}
 \STATE Calculate $\hat{S}_{t,i}: \hat{S}_{t,i} = \hat{S}_{t-1,i} + X_{t,i}$
 \ENDFOR
 \ENDFOR

\end{algorithmic}
\end{algorithm}

\section{Training Details}

\subsection{Tabular Experiments}
In tabular experiments on Leduc (\ref{fig:leduc_oracle} and \ref{fig:leduc_q}), we used the OpenSpiel library~\cite{lanctot2019openspiel}. 

In both experiments, we find the restricted player's policy using EXP3 as the regret minimization method. We used the EXP3.P variant of EXP3 \citep{bubeck2012regret}. Each iteration, we set $\gamma$ to $min(1, \frac{\sqrt{k * \log k}}{(e-1) * n})$ where $n$ is an estimated upper bound on the total cumulative regret upon completion of the algorithm.

Each iteration of the RM-BR DO and APSRO algorithms consisted of many EXP3 updates and BR updates (i.e. finding an exact BR for the former, and running episodes of Q-learning for the latter). In the Leduc RM-BR DO experiments (\ref{fig:leduc_oracle}), each iteration consisted of 100,000 EXP3 updates and 100 BR updates (1 BR update per 1,000 EXP3 updates). In the tabular Leduc APSRO experiments (\ref{fig:leduc_q}), each iteration consisted of 50,000 EXP3 updates and 500,000 tabular Q-learning episodes. EXP3 updates were run in batches of 10, and Q-learning episodes were run in batches of 100.  

For APSRO, we used OpenSpiel's Python implementation of a tabular epsilon-greedy Q-learner. All hyperparameters were default: step size~$= 0.1$ and epsilon~$= 0.2$.

To determine the EXP3 reward for the restricted player's sampled strategy, we find the exact expected payoff from playing that strategy against the current opponent BR, by traversing the game tree.

\subsection{Neural Experiments}
For neural experiments, we use the same RL best response hyperparameters in both PSRO and APSRO as well as in the measurement of approximate exploitability. When RL best responses are calculated, an independent best response learning process is performed for each player. RL hyperparameters for each game were selected based on sample efficiency and final performance against a fixed opponent.

New entries in the PSRO empirical payoff table are calculated by playing 3000, 3000, and 1000 evaluation episodes for each pair of population policies in Leduc, Goofspiel, and the 2D Cont. Action Hill-Climbing Game respectively. These episodes are not counted as experience collected for PSRO.

In deep RL experiments, we use the Multiplicate Weights Update (MWU) algorithm as our no-regret algorithm for APSRO. For all games, we perform a single metasolver update on every 10th iteration of our RL best response algorithm. We use a learning rate of 0.1. MWU expected payoffs for each population strategy are estimated using the average return over the last 1000 episodes in which that given population strategy was played. When measuring approximate exploitability, we measure a best response's payoff against the time-average MWU policy from a single APSRO iteration.

Our deep RL code was built on top of the RLlib framework \cite{rllib}, and any hyperparameters not specified are the version 1.0.1 defaults.

\begin{table}[H]
\centering
\begin{tabular}{ll}
Algorithm & DDQN \cite{van2016deep} \\
circular replay buffer size & 2e5 \\
prioritized experience replay & No \\
total rollout experience gathered each iter & 1024 steps \\
learning rate & 0.01 \\
batch size & 4096 \\
TD-error loss type & MSE \\
target network update frequency & every 10,000 steps \\
MLP layer sizes & [128, 128] \\
Activation function & ReLu \\
best response stopping condition & 1e5 episodes\\
\end{tabular}
\caption{Leduc Deep RL BR hyperparameters}
\label{table:ddqn-leduc}
\end{table}

\begin{table}[H]
\centering
\begin{tabular}{ll}
Algorithm & DDQN \cite{van2016deep} \\
circular replay buffer size & 2e5 \\
prioritized experience replay & No \\
total rollout experience gathered each iter & 1024 steps \\
learning rate & 0.002 \\
batch size & 2048 \\
TD-error loss type & MSE \\
target network update frequency & every 1e5 steps \\
MLP layer sizes & [128, 128, 128] \\
Activation function & ReLu \\
best response stopping condition & 1.5e6 timesteps\\
\end{tabular}
\caption{Goofspiel Deep RL BR hyperparameters}
\label{table:ddqn-goofspiel}
\end{table}

\begin{table}[H]
\centering
\begin{tabular}{ll}
Algorithm & PPO \cite{schulman2017proximal} \\
GAE $\lambda$ & 1.0 \\
entropy coeff & 0.01 \\
clip param & 0.2 \\
KL target & 0.01 \\
learning rate & 5e-4 \\
train batch size & 2048 \\
SGD minibatch size & 256 \\
num SGD epochs on each train batch & 30 \\
separate policy and value networks & Yes \\
continuous action range & [-1.0, 1.0] for each dim \\
MLP layer sizes & [32, 32] \\
Activation function & Tanh \\
best response stopping condition & 1e5 episodes\\
\end{tabular}
\label{Tab:ppo-params}
\caption{2D Continuous-Action Hill Climbing Game Deep RL BR hyperparameters}
\end{table}

\section{Experiment Code}
A GitHub link for our experiment code will be made available with the deanonymized version of this work.

\end{document}